\documentclass[letterpaper,11pt]{article}
\usepackage{amsmath, amsthm, amssymb, amstext, comment, graphicx, fullpage}
\usepackage{tikz-qtree}
\usetikzlibrary{shapes} 
\usetikzlibrary{arrows}
\usepackage{xcolor}
\usepackage{lipsum}

\usepackage[linesnumbered,ruled,vlined,boxed]{algorithm2e}
\usepackage{amsmath,amssymb}
\usepackage{verbatim}
\usepackage{enumerate}

\newtheorem{corollary}{Corollary}

\newtheorem{theorem}{Theorem}
\newtheorem{lemma}{Lemma}

\newcommand{\rank}{\text{rank}}
\renewcommand{\vec}{\mathbf}
\newcommand{\vx}{\vec{x}}

\newcommand{\cost}{\text{cost}}
\newcommand{\mac}{\text{mc}}
\newcommand{\soc}{\text{sc}}
\newcommand{\diam}{\text{diam}}

\title{Verifiably Truthful Mechanisms\footnote{The authors would like to thank Aris Filos-Ratsikas for a helpful discussion on characterizations of mechanisms for single peaked preferences and
Joan Feigenbaum, Kevin Leyton-Brown, Peter Bro Miltersen, and Tuomas Sandholm for useful feedback.}}

\author{
\textbf{Simina Br\^anzei}\\
\small{Aarhus University}\\
\small{simina@cs.au.dk}
\footnote{
Br\^anzei acknowledges support from the Sino-Danish Center for the Theory of Interactive Computation, funded by the Danish National Research Foundation and the National Science Foundation of China (under the grant 61061130540), and from the Center for research in the Foundations of Electronic Markets (CFEM), supported by the Danish Strategic Research Council.}
\and
\textbf{Ariel D. Procaccia}\\
\small{Carnegie Mellon University}\\
\small{arielpro@cs.cmu.edu}
\footnote{Procaccia was partially supported by the NSF under grants CCF-1215883 and IIS-1350598.}
}
\date{}

\begin{document}
\maketitle

\thispagestyle{empty}
\setcounter{page}{0}

\begin{abstract}
It is typically expected that if a mechanism is truthful, then the agents would, indeed, truthfully report their private information. But why would an agent \emph{believe} that the mechanism is truthful? We wish to design truthful mechanisms, whose truthfulness can be verified efficiently (in the computational sense). Our approach involves three steps: (i) specifying the structure of mechanisms, (ii) constructing a verification algorithm, and (iii) measuring the quality of verifiably truthful mechanisms. We demonstrate this approach using a case study: approximate mechanism design without money for facility location.
\end{abstract}

\newpage

\section{Introduction}

The mechanism design literature includes a vast collection of clever schemes that, in most cases, provably give rise to a specified set of properties. Arguably, the most sought-after property is \emph{truthfulness}, more formally known as \emph{incentive compatibility} or \emph{strategyproofness}: an agent must not be able to benefit from dishonestly revealing its private information. Truthfulness is, in a sense, a prerequisite for achieving other theoretical guarantees, because without it the mechanism may receive unpredictable input information that has little to do with reality. For example, if the designer's goal is to maximize utilitarian social welfare (the sum of agents' utilities for the outcome), but the mechanism is not truthful, the mechanism would indeed maximize social welfare --- albeit, presumably, with respect to the wrong utility functions!

An implicit assumption underlying the preceding (rather standard) reasoning is that when a truthful mechanism is used, (rational) agents would participate truthfully. This requires the agents to \emph{believe} that the mechanism is actually truthful. Why would this be the case? Well, in principle the agents can look up the proof of truthfulness.\footnote{A related, interesting question is: If we told human players that a non-truthful mechanism is provably truthful, would they play truthfully?}  A more viable option is to directly verify truthfulness by examining the specification of the mechanism itself, but, from a computational complexity viewpoint, this problem would typically be extremely hard --- even undecidable. This observation is related to the more general principle that the mechanism should be transparent or
simple, so that bounded rational economic agents can reason about it and take decisions efficiently.

Motivated by the preceding arguments, our goal in this paper is to design mechanisms that are \emph{verifiably truthful}. Specifically, we would like the verification to be \emph{efficient} --- in the computational sense (i.e., polynomial time), not the economic sense. In other words, the mechanism must be truthful, and, moreover, each agent must be able to efficiently verify this fact. 

\subsection{Our Approach and Results}
\label{subsec:res}

Our approach to the design of verifiably truthful mechanisms involves three steps:

\begin{enumerate}
\item[I] \emph{Specifying the structure of mechanisms}: The verification algorithm will receive a mechanism as input --- so we must rigorously specify which mechanisms are admissible as input, and what they look like.

\item[II]\emph{Constructing a verification algorithm}: Given a mechanism in the specified format, the algorithm decides whether the mechanism is truthful. 

\item[III]\emph{Measuring the quality of verifiably truthful mechanisms}: The whole endeavor is worthwhile (if and) only if the family of mechanisms whose truthfulness can be verified efficiently (via the algorithm of Step 2) is rich enough to provide high-quality outcomes. 
\end{enumerate}

We instantiate this program in the context of a case study: \emph{approximate mechanism design without money for facility location}~\cite{PT13ec}. The reason for choosing this specific domain is twofold. First, a slew of recent papers has brought about a good understanding of what quality guarantees are achievable via truthful facility location mechanisms~\cite{AFPT10,LWZ09,LSWZ10,NST12,FT10,FT13a,FT13b,Than10,TIY11,CYZ13,WF13}. Second, facility location has also served as a proof of concept for the approximate mechanism design without money agenda~\cite{PT13ec}, whose principles were subsequently applied to a variety of other domains, including allocation problems~\cite{GCR09,GC10,DG10,CGG13}, approval voting~\cite{AFPT11}, kidney exchange~\cite{AFKP14,CFP11}, and scheduling~\cite{Kout11}. Similarly, facility location serves as an effective proof of concept for the idea of verifiably truthful mechanisms, which, we believe, is widely applicable. 

We present our results according to the three steps listed above: 

\begin{enumerate}
\item[I] In \S\ref{sec:step1}, we put forward a representation of facility location mechanisms. In general, these are arbitrary functions mapping the reported locations of $n$ agents on the real line to the facility location (also on the real line). We present \emph{deterministic} mechanisms as \emph{decision trees}, which branch on comparison queries in internal nodes, and return a function that is a convex combination of the reported locations in the leaves. Roughly speaking, randomized mechanisms are distributions over deterministic mechanisms, but we use a slightly more expressive model to enable a concise representation for certain randomized mechanisms that would otherwise need a huge representation. 

\item[II] The \emph{cost} of an agent is the distance between its (actual) location, which is its private information, and the facility location. A deterministic mechanism is truthful if an agent can never decrease its cost by reporting a false location. In \S\ref{sec:step2}, we show that the truthfulness of a deterministic mechanism can be verified in polynomial time in the size of its decision tree representation and number of agents. We also demonstrate that one cannot do much better: it is necessary to at least inspect all the tree's leaves. We establish that the efficient verification result extends to randomized mechanisms, as long as the notion of truthfulness is \emph{universal truthfulness}: it must be impossible to gain from manipulating one's reported location, regardless of the mechanism's coin tosses. 

\item[III] Building on the results of Step II, we focus on decision trees of polynomial size --- if such mechanisms are truthful, their truthfulness can be efficiently verified. In \S\ref{sec:step3}, we study the quality of polynomial-size decision trees, via two measures of quality: the \emph{social cost} (the sum of agents' cost functions) and the maximum cost (of any agent). Figure~\ref{tab:summary} summarizes our results. The table on the left shows tight bounds on the (multiplicative, worst-case) approximation ratio that can be achieved by truthful mechanisms~\cite{PT13ec} --- deterministic in the first row, randomized in the second. The (lower) bound for the maximum cost of universally truthful mechanisms is new.
%%%\footnote{The unavailable cell is due to the fact that Procaccia and Tennenholtz focused on truthfulness in expectation; see Section~\ref{sec:disc} for a discussion of this point.}
%%\footnote{The lower bound on general universally truthful mechanisms with respect to the maximum cost is shown in this paper}.
The results for efficiently verifiable mechanisms are shown in the right table. Our main results pertain to the social cost (left column): while deterministic polynomial-size decision trees only achieve an approximation ratio of $\Theta(n/\log n)$, we construct (for any constant $\epsilon>0$) a polynomial-size, randomized, universally truthful decision tree approximating the social cost to a factor of $1+\epsilon$. 
\end{enumerate}

\begin{figure}
\begin{center}
\begin{tabular}{|c||c|c|}
\hline
\begin{tabular}{@{}c@{}}\small{\emph{General}} \\ \small{\emph{Mechanisms}}\end{tabular}& \small{Social Cost}& \small{Max Cost}\\
\hline\hline
\small{Truthful} & 1 & 2\\
\hline\hline
\small{Univ. Truthful} & 1 & 2 $(*)$\\
\hline
\end{tabular}
\quad
\begin{tabular}{|c||c|c|}
\hline
\begin{tabular}{@{}c@{}}\small{\emph{Polynomial-size}} \\ \small{\emph{Decision Trees}}\end{tabular}& \small{Social Cost}& \small{Max Cost}\\
\hline\hline
\small{Truthful} & $\Theta\left(\frac{n}{\log{n}}\right)$ & 2\\
\hline\hline
\small{Univ. Truthful} & $1 + \epsilon$ & 2\\
\hline
\end{tabular}
\end{center}
\caption{The results of \S\ref{sec:step3}, outlined in \S\ref{subsec:res}. The lower bound $(*)$ for general mechanisms is also shown in this paper.}
\label{tab:summary}
\end{figure}

\subsection{Related Work}

Verification is a common theme in algorithmic mechanism design, but in the past it was always the agents' reports that were being verified, not the properties of the mechanism itself. In fact, in the eponymous paper by Nisan and Ronen~\cite{NR01}, a class of mechanisms with verification (and money) for scheduling was proposed. These mechanisms are allowed to observe both the reported types and the actual types (based on the execution of jobs), and payments may depend on both. Verification of agents' reports has subsequently played a role in a number of papers; of special note is the work of Caragiannis et al.~\cite{CESY12}, who focused on different notions of verification. They distinguished between \emph{partial verification}, which restricts agents to reporting a subset of types that is a function of their true type (e.g., in scheduling agents can only report that they are slower than they actually are, not faster), and \emph{probabilistic verification}, which catches an agent red handed with probability that depends on its true type and reported type. There are also examples of this flavor of verification in approximate mechanism design without money~\cite{Kout11}.

A small body of work in multiagent systems~\cite{PW03,BFVW06,TGV09} actually aims to verify properties of mechanisms and games. The work of Tadjouddine et al.~\cite{TGV09} is perhaps closest to ours, as they verify the truthfulness of auction mechanisms. Focusing on the Vickrey Auction~\cite{Vick61}, they specify it using the Promela process modeling language, and then verify its truthfulness via model checking techniques. This basically amounts to checking all possible bid vectors and deviations in a discretized bid space. To improve the prohibitive running time, abstract model checking techniques are applied. While model checking approaches are quite natural, they inevitably rely on heuristic solutions to problems that are generally very hard. In contrast, we are interested in mechanisms whose truthfulness can be verified \emph{in polynomial time}.

%%%%%%%%%%
%%%%%%%%%% SIMINA: PARAGRAPH ABOUT THE MUALEM PAPER %%%%
Mu'alem~\cite{Mualem05} considers a motivating scenario similar to ours and focuses on testing extended monotonicity, which is a property required for truthfulness in the single parameter domain studied therein. In particular, Mu'alem shows that if a function $f$ is $\epsilon$-close 
to extended monotonicity, then there \emph{exists} an associated payment function $p$ such that the mechanism given by the tuple $(f, p)$ is $(1- 2 \epsilon$)-truthful. She also describes a shifting technique for obtaining almost truthful mechanisms and a monotonicity tester.
While studying truthfulness in the context of property testing remains an interesting question for future work, we would like to obtain mechanisms whose truthfulness can be verified exactly and in polynomial time (independent of the size of the domain --- in fact, our domain is continuous).
On a technical level, we study a setting without payments, so our setting does not admit a close connection between monotonicity and truthfulness. 

%%%%%%%%%%

Kang and Parkes~\cite{KP06} consider the scenario in which multiple entities (e.g. companies, people, network services) can deploy mechanisms in an open computational infrastructure. Like us, they are interested in verifying the truthfulness of mechanisms, but they sidestep the question of how mechanisms are represented by focusing on what they call \emph{passive verification}: their verifier acts as an intermediary and monitors the sequence of inputs and outputs of the mechanism. The verifier is required to be sound and complete; in particular, if the mechanism is not strategyproof, the verifier is guaranteed to establish this fact after observing all the possible inputs and outputs.  

Our work is also related to the line of work on \emph{automated mechanism design}~\cite{CS02b}, which seeks to automatically design truthful mechanisms that maximize an objective function, given a prior distribution over agents' types. In an informal sense, this problem is much more difficult than our verification problem, and, indeed, in general it is computationally hard even when the mechanism is explicitly represented as a function whose domain is all possible type vectors. Automated mechanism design is tractable in special cases --- such as when the number of agents is constant and the mechanism is randomized --- but these results do not yield nontrivial insights on the design of verifiably truthful mechanisms.

\section{Step I: Specifying the Structure of Mechanisms}
\label{sec:step1}

We consider the (game-theoretic) facility location problem~\cite{PT13ec}. An instance includes a set $N=\{1,\ldots,n\}$ of agents. Each agent $i\in N$ has a location $x_i$. 
%which is $i$'s private information. 
The vector $\vx = \langle x_1, \ldots, x_n \rangle$ represents the location profile. We relegate the presentation of the strategic aspects of this setting to Section~\ref{sec:step2}. 

\subsection{Deterministic Mechanisms}

A \emph{deterministic mechanism} (for $n$ agents) is a function $\mathcal{M}: \mathbb{R}^{n} \rightarrow \mathbb{R}$, which maps each location profile $\vx$ to a facility location $y\in \mathbb{R}$. We put forward a simple, yet expressive, representation of deterministic mechanisms, via \emph{decision trees}. 

In more detail, given input $\vx = \langle x_1, \ldots, x_n \rangle$,
the mechanism is represented as a tree, with: 
\begin{itemize}
\item \emph{Internal nodes}: used to verify sets of constraints over the input variables. We focus on a comparison-based model of computation, in which each internal node 
%$\mathcal{N}$ contains exactly one or zero constraints: % (if $\mathcal{N}$ has zero constraints, then the decision taken in $\mathcal{N}$ is by default true).
%
%\begin{itemize}
%
%\item[(a)] If $|\mathcal{N}| = 1$, The node 
verifies one constraint, of the form $(x_i \geq x_j)$, $(x_i \leq x_j)$, $(x_i > x_j)$, or $(x_i < x_j)$, for some $i, j \in N$.
The node has two outgoing edges, that are taken depending on whether the condition is true or false.

%\item[(b)] If $|\mathcal{N}| = 0$, the node is empty and can be used to take randomized decisions. An empty node can have any number of outgoing edges, each labeled with the probability
%that the execution of the mechanism will continue down on that path. The probabilities on the outgoing edges from every empty node sum up to $1$.
%For example, if from some state the mechanism chooses one of $n$ possible paths with probability $1/n$ each, then the node $\mathcal{N}$ corresponding
%to that state will have $n$ outgoing edges, each labeled with the execution probability ($1/n$).
%\end{itemize}

\item \emph{Leaves}: store the outcome of the mechanism if the path to that leaf is taken, i.e. the facility location.% $y$. 
We require that for each leaf $\mathcal{L}$, the location of the facility at $\mathcal{L}$, $y_{\mathcal{L}}(x)$, is a convex combination of the input locations: 
$y_{\mathcal{L}}(\vx) = \sum_{i=1}^{n} \lambda_{\mathcal{L},i} \cdot x_i$, where the $\lambda_{\mathcal{L},i}$ are constants with $\lambda_{\mathcal{L},i} \geq 0$ and $\sum_{i=1}^{n} \lambda_{\mathcal{L},i} = 1$.
%
%We require that for each leaf $\mathcal{L}$, $y$ is a convex combination of the input locations; i.e., $y = \sum_{i=1}^{n} \lambda_i \cdot x_i$, where the $\lambda_{i}$ are constants with $\lambda_i \geq 0$ and $\sum_{i=1}^{n} \lambda_i = 1$.
\end{itemize}

\tikzset{
  treenode/.style = {align=center, inner sep=2.4pt,
  font=\sffamily},
  arn_m/.style = {treenode, rectangle, rounded corners=1.5mm, black,level distance=1.1cm,sibling distance=1.8cm,
  font=\sffamily\bfseries, draw=black, fill=white!90!gray, fill opacity = 0.9, minimum width=0.7em,minimum height=1.8em},
  arn_chance/.style = {treenode, ellipse, white,level distance=1.5cm,sibling distance=2.5cm, minimum width=1.2em, 
minimum height=1.2em, draw=black, fill=black,very thick},
  arn_n/.style = {treenode, rectangle, white, font=\sffamily\bfseries, 
  draw=gray, fill=gray, minimum width=2em},
%  arn_r/.style = {treenode, rectangle, red, draw=red, 
%  minimum width=3em, very thick},
  arn_r/.style = {treenode, rectangle, black,draw=black, fill=gray!70!white, 
  minimum height=2.2em,rounded corners=2mm,
  minimum width=2.3em, thick},
  arn_x/.style = {treenode, rectangle, draw=black,
  minimum width=0.5em, minimum height=0.5em}
}

\begin{figure}[b]
\begin{minipage}[b]{0.50\linewidth}
\centering
\scalebox{.74}{
\begin{tikzpicture}[->,>=stealth',level/.style={sibling distance = 5.2cm/#1,
  level distance = 1.8cm}] 
\node [arn_r] {$\mathtt{\frac{x_1 + x_2 + \ldots x_n}{n}}$}
;
\end{tikzpicture}
}
\caption{The average mechanism.} \label{fig:Average_N}
\label{fig:avg}
\end{minipage}
\hspace{0.5cm}
\begin{minipage}[b]{0.50\linewidth}
\centering
\scalebox{.74}{
\begin{tikzpicture}[->,>=stealth',level/.style={sibling distance = 6cm/#1,
  level distance = 1.6cm}] 
\node [arn_r] {$\mathtt{x_i}$}
;
\end{tikzpicture}
}
\caption{Dictatorship of agent $i$.} \label{fig:Dictator}
\label{fig:dictator}
\end{minipage}
\end{figure}

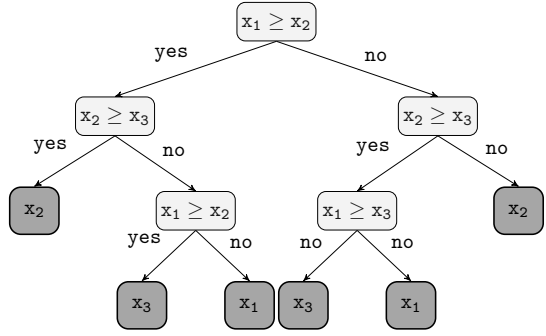
\begin{figure}[t]
\centering
\scalebox{.74}{%\input{   %plot.tex}}
\begin{tikzpicture}[->,>=stealth',level/.style={sibling distance = 5.8cm/#1,
  level distance = 1.7cm}] 
\node [arn_m] {$\mathtt{x_1 \geq x_2}$} 
child{ node [arn_m] {$\mathtt{x_2 \geq x_3}$}
            child{ node [arn_r] {$\mathtt{x_2}$} 
 edge from parent node[above left] {$\mathtt{yes}$}            
}
            child{ node [arn_m] {$\mathtt{x_1\geq x_2}$}
                            child{ node [arn_r] {$\mathtt{x_3}$}  edge from parent node[above left]
                         {$\mathtt{yes}$}}
                            child{ node [arn_r] {$\mathtt{x_1}$}  edge from parent node[above right]
                         {$\mathtt{no}$}}
 edge from parent node[above right]
                         {$\mathtt{no}$}            
    }
 edge from parent node[above left]
                         {$\mathtt{yes}$}  
}
    child{ node [arn_m] {$\mathtt{x_2 \geq x_3}$}
            child{ node [arn_m] {$\mathtt{x_1 \geq x_3}$} 
                            child{ node [arn_r] {$\mathtt{x_3}$}  edge from parent node[above left]
                         {$\mathtt{no}$}}
                            child{ node [arn_r] {$\mathtt{x_1}$}  edge from parent node[above right]
                         {$\mathtt{no}$}}
 edge from parent node[above left]
                         {$\mathtt{yes}$}            
}          
            child{ node [arn_r] {$\mathtt{x_2}$}  edge from parent node[above right]
                         {$\mathtt{no}$}            
}
edge from parent node[above right]
                         {$\mathtt{no}$}
        }
;
\end{tikzpicture}
}%}
\caption{The median mechanism for 3 agents.} \label{fig:Median3}
\end{figure}

For example, Figure~\ref{fig:avg} shows the decision tree representation of the average mechanism, which returns the average of the reported locations. It is just a single leaf, with coefficients $\lambda_i=1/n$ for all $i\in N$. Figure~\ref{fig:dictator} shows a dictatorship of agent $i$ --- whatever location is reported by agent $i$ is always selected. Figure~\ref{fig:Median3} shows the median mechanism for $n=3$, which returns the median of the three reported locations; this mechanism will play a key role later on. 

We remark that our positive results are based on mechanisms that have the so-called \emph{peaks-only} property: they always select one of the reported locations. However, our more expressive definition of the leaves of the decision tree (as convex combinations of points in $\vx$) is needed to compute optimal solutions under one of our two objectives (as we discuss below), and also strengthens our negative results.

\subsection{Randomized Mechanisms}

Intuitively, randomized mechanisms are allowed to make branching decisions based on coin tosses. Without loss of generality, we can just toss all possible coins in advance, so a randomized mechanism can be represented as a probability distribution over deterministic decision trees. 
However, this can lead to a large representation of simple mechanisms that consist of the same (fixed) subroutine executed with possibly different input variables.
% on every execution path a fixed subroutine with possibly different 
% input variables. 
For example, the mechanism that selects a (not very small) subset of agents uniformly at random and computes the median of the subset can be seen as a median mechanism parameterized by the identities of the agents.
In order to be able to represent such mechanisms concisely, we make the representation a bit more expressive.

Formally, a randomized mechanism is represented by a decision tree with a chance node of degree $K$ as the root,
such that the $r$'th edge 
selects 
a decision tree $\mathcal{T}_r$
and is taken with probability $p_r$,
where $\sum_{i=1}^{K} p_i = 1$.
%$p_1, \ldots, p_k$ one of the functions, $\mathcal{F}_1, \ldots, \mathcal{F}_k$, respectively, where $p_1 + \ldots p_k = 1$.
Each tree $\mathcal{T}_r$ is defined as follows:
\begin{itemize}
\item There is a set of agents $N_r\subseteq N$, such that the locations $x_i$ for $i\in N$ appear directly in the internal nodes and leaves of the tree. 
\item There is a set of parameters $Z_r = \{z_{r,1},\ldots,z_{r,m_r}\}$, that also appear in the internal nodes and leaves of $\mathcal{T}_r$, where $0 \leq m_r\leq |N\setminus N_r|$.
\item The description of $\mathcal{T}_r$ includes a probability distribution over tuples of $m_r$ distinct agents from $N\setminus N_r$.   
\end{itemize} 
%It has a possibly set of arguments $Z_i = \{z_{i,1}, \ldots, z_{i,n_i}\}$ that are bound at execution time 
%to distinct locations from a domain $X_i = \{x_{i,1}, \ldots, x_{i,n_i}\} \subseteq X$.
%%by mapping each argument $z_{i,j}$ to a distinct agent location from some domain
%%$X_i = \{x_{i,1}, \ldots, x_{i,n_i}\}$, where $n_i \in [n]$. 
%\item
%$\mathcal{F}_i$ is computed using a deterministic decision tree $\mathcal{T}_i$, which uses both (symbolic) variables from $Z_i$ 
%and fixed variables from $X \setminus X_i$.
%\end{itemize}

The semantics are as follows. At the beginning of the execution, a die is tossed to determine the index $r \in \{1, \ldots, K\}$ of the function (i.e. tree $\mathcal{T}_r$) to be implemented. Then, the parameters $z_{r,j}$ are bound to locations of agents from $N\setminus N_r$ according to the given probability distribution for $\mathcal{T}_r$; each $z_{r,j}$ is bound to a different agent. 
At this point all the parameters in the nodes and leaves of $\mathcal{T}_r$ have been replaced by variables $x_i$,
and we just have a deterministic decision tree, which %, in which all the locations are variables $x_i$, which 
is executed as described above.

For example, say we want to implement the mechanism that selects three agents uniformly at random from $N$ and outputs the median of these three agents. This mechanism requires a randomized decision tree with a chance node of degree one, that selects
with probability $p_1 = 1$
a single decision tree $\mathcal{T}_1$, which is the tree in Figure~\ref{fig:Median3} with the $x_i$ variables replaced by $z_i$.
We set $N_1=\emptyset$ (thus the tree $\mathcal{T}_1$ is completely parameterized), and the probability distribution over distinct subsets of size $3$ from $N\setminus N_1=N$ is just the uniform distribution over such subsets.

\section{Step II: Constructing a Verification Algorithm}
\label{sec:step2}

In Section~\ref{sec:step1} we focused on the non-strategic aspects of the facility location game: agents report their locations, which are mapped by a mechanism to a facility location. The potential for strategic behavior stems from the assumption that the agents' locations $\vx$ are private information --- $x_i$ represents agent $i$'s \emph{ideal} location for the facility (also known as agent $i$'s \emph{peak}). Like Procaccia and Tennenholtz~\cite{PT13ec}, and almost all subsequent papers, we assume that the \emph{cost} of agent $i$ for facility location $y$ is simply the Euclidean distance between (the true) $x_i$ and $y$,
$$
\cost(x_i,y) = |x_i - y|.
$$

\subsection{Deterministic Mechanisms}

A deterministic mechanism $\mathcal{M}:\mathbb{R}^n\rightarrow \mathbb{R}$ is \emph{truthful} if for every location profile $\vx\in\mathbb{R}^n$, every agent $k\in N$, and every $x_k'\in\mathbb{R}$, $\cost(x_k,\mathcal{M}(\vx))\leq \cost(x_k,\mathcal{M}(x_k',\vx_{-k})$, where $\vx_{-k}=\langle x_1,\ldots,x_{k-1},x_{k+1},\ldots,x_n\rangle$. Our next goal is to construct an algorithm that receives as input a deterministic mechanism, represented as a decision tree, and verifies that it is truthful. 

The verification algorithm is quite intuitive, although its formal specification is somewhat elaborate. Consider a mechanism $\mathcal{M}:\mathbb{R}^n\rightarrow \mathbb{R}$ that is represented by a tree $\mathcal{T}$. For a leaf $\mathcal{L}$, denote the location chosen by $\mathcal{M}$ at this leaf by $y_{\mathcal{L}}(\vx) = \sum_{i=1}^{n} \lambda_{\mathcal{L},i} \cdot x_{i}$.
In addition, let $\mathcal{C(L)}$ denote the set of constraints encountered on the path to $\mathcal{L}$. For example, the set of constraints corresponding to the leftmost 
leaf in Figure \ref{fig:Median3} is $\{(x_1 \geq x_2), (x_2 \geq x_3)\}$, while the second leaf from the left verifies: $\{(x_1 \geq x_2),(x_2 < x_3), (x_1 \geq x_3)\}$. We define a procedure, \textsc{Build-Leaf-Constraints}, that gathers these constraints (Algorithm~\ref{alg:leaf}). One subtlety is that the procedure ``inflates'' strict inequality constraints to constraints that require a difference of at least $1$; we will explain shortly why this is without loss of generality.

\begin{algorithm}[t]
\label{alg:main}
\caption{\textsc{Truthful$(\mathcal{T})$} //\textit{ Verifier for Deterministic Mechanisms}}
\KwData{mechanism $\mathcal{T}$}
\KwResult{\emph{true} if $\mathcal{T}$ represents a truthful mechanism, \emph{false} otherwise}
  \textsc{Build-Leaf-Constraints($\mathcal{T}$)}\\
\ForEach{$k \in N$} {
    \ForEach{leaf $\mathcal{L} \in \mathcal{T}$} {
	//\textit{ $y_{\mathcal{L}}(\vx)$ is the symbolic expression for the facility at $\mathcal{L}$ on $\vx$}\\
	// \textit{$d_k(\vx)$ is agent $k$'s distance from the facility on input $\vx$}\\
% $y_{\mathcal{L}}(\vx)$}\\
%denote $d_k(\vx) = |x_k - y_{\mathcal{L}}(\vx)|$}\\
        \ForEach{$d_k(\vx) \in \left\{ x_k - y_{\mathcal{L}}(\vx),  -x_k + y_{\mathcal{L}}(\vx)\right\}$} {
	//\textit{ two cases, for $x_k$ to the left or right of the facility $y_{\mathcal{L}}(\vx)$}\\
        \ForEach{leaf $\mathcal{L}' \in \mathcal{T}$} {
        \ForEach{$d'_k(\vx') \in \left\{ x_k' - y_{\mathcal{L'}}(\vx'),  -x_k' + y_{\mathcal{L'}}(\vx')\right\}$} {
%          $x \leftarrow (x_i, x_{-i})$\\
%          $x' \leftarrow (x_i', x_{-i})$\\
%          $\epsilon \leftarrow Solve(i, x, C_{\mathcal{L}}, x', C_{\mathcal{L}'})$\\ 
          %%%Cases I need to consider here: 
% (a) $x_i >= y_L and x_i' >= y_L'
% (b) $x_i >= y_L and x_i' < y_L'
% (c) $x_i < y_L and x_i' >= y_L'
% (d) $x_i < y_L and x_i' < y_L'
%%% // Build symbolic expressions for the locations 
          $\mbox{\emph{inc}}(\vx, \vx') \leftarrow \left\{(d_k(\vx) - d'_k(\vx') \geq 1),d_k(\vx)\geq 0,d_k'(\vx')\geq 0\right\}$ \\//\textit{ utility increase from $\vx$ to $\vx'$, distances are non-negative}\\
          \If{\textsc{Exists-Solution}$\left(k, \mathcal{C}_{\mathcal{L}}, \mathcal{C}_{\mathcal{L}'}, \mbox{\textit{inc}}\right)$} {
              \Return \emph{False}\\
          } 
}
}
        }
    }
}

\Return \emph{True}

\end{algorithm}

%%%%%%%%%%%%%%SIMINA:::  NOTE ::: MOVED THE PSEUDOCODE OF BUILD-LEAF-CONSTRAINTS TO THE APPENDIX %%%%%%%%%%%%%%%%%
%%%%%%%%%%%%%%%%%%%%%%%%%%%%%%%%%%%%%%%%%%%%%%%%%%%%%%%%%%%%%%

%%%%%%%%%%%%%%%%%%%%%%%%%%%%%%%%%%%%%%%%%%%%%%%%%%%%%%
%%%%%%%%%%%%%%%%%%%%%%%%%%%%%%%%%%%%%%%%%%%%%%%%%%%%%%
\begin{algorithm}[t]
\caption{\textsc{Exists-Solution}($k, \mathcal{C}_{\mathcal{L}}, \mathcal{C}_{\mathcal{L'}}', \mbox{\textit{inc}})$}
\KwData{agent $k$ and symbolic constraint sets $\mathcal{C}_{\mathcal{L}}, \mathcal{C}_{\mathcal{L}'}'$, \textit{inc}}
\KwResult{\emph{true} $\iff$ $\exists \; x_1, \ldots, x_n, x_k' \in \mathbb{R}^{+}$ \emph{subject to} $\mathcal{C}_{\mathcal{L}}(\vx) \; \& \; \mathcal{C}_{\mathcal{L}}'(x_k', \vx_{-k}) \; \& \;  \mbox{\textit{inc}}(\vx, (x_k', \vx_{-k}))$}

%$x_1, \ldots, x_n, x_i' \geq 0$
%$x_{-i} = x_{-i}'$,
%$x \geq 0, x' \geq 0$, and
%$\Delta(x,x')$ holds (i.e. player $i$ has improving deviation $x_i'$ from $x$)\\
$\vx' \leftarrow (x_1, \ldots, x_{i-1}, x_i', x_{i+1}, \ldots, x_n)$\\

$W \leftarrow \{\mathcal{C}_{\mathcal{L}}(\vx), \mathcal{C}_{\mathcal{L}}'(\vx'), \mbox{\textit{inc}}(\vx, \vx')\}$ \\
$\vec{z} \leftarrow (x_1, \ldots, x_n, x_i')$\\

\Return $\mbox{\textsc{Solve}}(\vec{z}, W, \vec{z} \geq 0)$\ //\ \emph{Linear program solver}
%
%\textbf{return}: \emph{Solve}$\left(\begin{array}{c}
%\left\langle x, x' \right\rangle,
%\left\{ \mathcal{C}(\vx), \mathcal{C}'(x'), \bigcup_{k \in N \setminus \{i\}} \{ x_k - x_k' \geq 0, x_k' - x_k \geq 0\}, \bigcup_{k \in N} \{x_k \geq 0, x_k' \geq 0\}, \Delta(x, x') \right\}
%\end{array}\right)$
\end{algorithm}
%%%%%%%%%%%%%%%%%%%%%%%%%%%%%%%%%%%%%%%%%%%%%%%%%%%%%%
%%%%%%%%%%%%%%%%%%%%%%%%%%%%%%%%%%%%%%%%%%%%%%%%%%%%%%
%%%%%%%%%%%%%%%%%%%%%%%%%%%%%%%%%%%%%%%%%%%%%%%%%%%%%%
%%%%%%%%%%%%%%%%%%%%%%%%%%%%%%%%%%%%%%%%%%%%%%%%%%%%%%

The main procedure, \textsc{Truthful} (given as Algorithm~\ref{alg:main}), checks whether there exist location profiles $\vx$ and $\vx'$ that differ only in the $k$'th coordinate, such that $\vx$ reaches leaf $\mathcal{L}$ (based on the constraints of the \textsc{Build-Leaf-Constraints} procedure, given
as Algorithm $3$ in Section~\ref{sec:missing_pseudocode} of the appendix), $\vx'$ reaches leaf $\mathcal{L}'$, and $\cost(x_k,y_{\mathcal{L}'}(\vx'))+1\leq \cost(x_k,y_{\mathcal{L}}(\vx))$, i.e. the reduction in cost is at least $1$.

So why can we ``inflate'' strict inequalities by requiring a difference of $1$? Assume that we are given a mechanism $\mathcal{T}$ and an agent $i$ such that for some strategy profiles $\vx$ and $\vx'$ with $\vx_{-i} = \vx_{-i}'$, agent $i$ can strictly benefit by switching from $\vx$ to $\vx'$.
Then there exists $\epsilon > 0$ such that agent $i$'s improvement is at least $\epsilon$, and for every strict inequality satisfied by $\vx$ and $\vx'$, the difference between the terms is at least $\epsilon$; for example, if $x_k > x_l$, then it is the case that $x_k - x_l \geq \epsilon$.
Since each facility location is a homogeneous linear function of the input $\vx$, all variables can be multiplied by $\frac{1}{\epsilon}$ to obtain that $\vx/\epsilon$ and $\vx'/\epsilon$ satisfy the more stringent constraints (with a difference of $1$) on agent locations and facility locations. 
%For example, if $x_k > x_l$, then $\frac{x_k}{\epsilon} - \frac{x_l}{\epsilon} \geq 1$. Then in order to find a profitable deviation for agent $i$, it is sufficient and necessary to show the existence of $\langle \vx, \vx' \rangle$ such that $x_k - x_k' \geq 0$ and $x_k' - x_k \geq 0$, the reduction in cost is at least $1$, and every strict comparison between the input values is mapped to one requiring a difference of at least $1$ between the terms (for example, $x_k > x_l$ becomes $x_k - x_l \geq 1$).

Finally, this algorithm works in polynomial time because the procedure \textsc{Exists-Solution}, which checks whether there is a solution to the different constraints (corresponding to a profitable manipulation), just solves a linear program using the procedure \textsc{Solve}. 

We summarize the preceding discussion with the following theorem. 

\begin{theorem}
\label{thm:step2ub}
Let $N=\{1,\ldots,n\}$. The truthfulness of a deterministic mechanism $\mathcal{M}$ represented as a decision tree $\mathcal{T}$ can be verified in polynomial time in $n$ and $|\mathcal{T}|$. 
\end{theorem}

Algorithm 1 
%%%NOTE: the number of the algorithm when using the reference was showing as 13, which is why I wrote it directly%%%%%%%%%%%%
essentially carries out a brute force search over pairs of leaves to find a profitable manipulation. Under the decision tree representation, is it possible to verify truthfulness much more efficiently? Our next result % gives a negative answer.
answers this question in the negative. The proof is included in the appendix, together with all the other proofs omitted from the main text.

\begin{theorem}
\label{thm:step2lb}
Let $N=\{1,\ldots,n\}$ with $n\geq 2$, and $\ell\leq n!$.
% \ in \left\{1, \ldots, 2^{\binom{n}{2}}\right\}$
Then any algorithm that verifies truthfulness for every deterministic decision tree with $\ell$ leaves for $n$ agents must inspect all the leaves in the worst case. 
\end{theorem}

Crucially, our decision trees are binary trees, so the number of leaves is exactly the number of internal nodes plus one. Theorem~\ref{thm:step2lb} therefore implies:

\begin{corollary}
Let $N=\{1,\ldots,n\}$, $n\geq 2$. Any verification algorithm requires superpolynomial time in $n$ (in the worst-case) to verify the truthfulness of trees of superpolynomial size in $n$. 
\end{corollary}

\subsection{Randomized Mechanisms}

In the context of randomized mechanisms, there are two common options for defining truthfulness: \emph{truthfulness in expectation} and \emph{universal truthfulness}. In our context, truthfulness in expectation means that an agent cannot decrease its expected distance to the facility by deviating; universal truthfulness means that the randomized mechanism is a probability distribution over truthful deterministic mechanisms, i.e., an agent cannot benefit from manipulation regardless of the mechanism's random coin tosses. Clearly, the former notion of truthfulness is weaker than the latter. In some settings, truthful-in-expectation mechanisms are known to achieve guarantees that cannot be obtained through universally truthful mechanisms~\cite{DD13}.

We focus on universal truthfulness, in part because we do not know whether truthful-in-expectation mechanisms can be efficiently verified (as we discuss in \S\ref{sec:disc}).
Using Theorem~\ref{thm:step2ub}, universal truthfulness is easy to verify, because it is sufficient and necessary to verify the truthfulness of each of the decision trees in the mechanism's support. One subtlety is the binding of agents in $N\setminus N_r$ to the $z_{r,j}$ parameters. However, for the purpose of verifying truthfulness, any binding will do by symmetry between the agents in $N\setminus N_r$. We therefore have the following result:

\begin{theorem}
Let $N=\{1,\ldots,n\}$. The universal truthfulness of a randomized mechanism $\mathcal{M}$ represented as a distribution over $K$ decision trees $\mathcal{T}_1,\ldots,\mathcal{T}_K$ can be verified in polynomial time in $n$ and its representation size, $\sum_{r=1}^K|\mathcal{T}_r|$. 
\end{theorem}

%\begin{proof}
%Let $\mathcal{C}$ be the chance node of $\mathcal{M}$, with degree $K$ and probabilities $p_1, \ldots, p_K$, where $\sum_{i=1}^{K} p_i = 1$.
%Mechanism $\mathcal{M}$ is universally truthful if and only if for all $i \in [K]$, the function $\mathcal{F}_i$ selected by $\mathcal{M}$
%on the $i$-th edge, is universally truthful.
%Observe that the binding of the arguments of $\mathcal{F}_i$ is arbitrary and moreover, the fixed locations in the representation of $\mathcal{F}_i$
%do not overlap with the domain of $Z_i$. Thus if some binding $z_{i,1} \leftarrow x_{j_1}$, $\ldots$, $z_{i,n_i} \leftarrow x_{j_{n_i}}$
%gives a truthful mechanism, then any other binding also gives a truthful mechanism. {\color{blue}Ariel do I need to explain more?}
%Thus truthfulness of $\mathcal{M}$ can be established by iterating over all the functions $\mathcal{F}_i$ and checking if a fixed arbitrary binding of 
%the arguments of $\mathcal{F}_i$ results in a truthful deterministic mechanism. By Theorem 1 ({\color{blue}:number?}), the deterministic
%verifier runs in $O(poly(n, |\mathcal{F}_i|))$, and so the overall runtime of the universal verifier is $O(poly(n, |\mathcal{M}|))$.
%\end{proof}

\section{Step III: Measuring the Quality of Verifiably Truthful Mechanisms}
\label{sec:step3}

We have shown that the truthfulness of mechanisms represented by decision trees of polynomial size can be verified in polynomial time. This result is encouraging, but it is only truly meaningful if decision trees of polynomial size can describe mechanisms that provide good guarantees with respect to the quality of the solution. 

Like Procaccia and Tennenholtz~\cite{PT13ec}, and subsequent papers, we measure solution quality in the facility location domain via two measures. 
The \emph{social cost} of a facility location $y\in \mathcal{R}$ for a location profile $\vx\in\mathbb{R}^n$ is
$$
\soc(\vx,y) = \sum_{i=1}^{n} \cost(x_i, y),
$$ 
and the \emph{maximum cost} is
$$
\mac(\vx,y)=\max_{i\in N} \cost(x_i,y).
$$
We denote the optimal solutions with respect to the social cost and maximum cost by 
$\soc^{*}(\vx) = \min_{y \in \mathbb{R}} \sum_{i=1}^{n} \cost(x_i, y)$, and $\mac^{*}(\vx) = \min_{y \in \mathbb{R}} \max_{i \in N} \cost(x_i, y)$, respectively.

%We begin our analysis by studying the approximation ratios of deterministic mechanisms.

\subsection{Deterministic Mechanisms}

Let us first review what can be done with deterministic mechanisms represented by decision trees of arbitrary size, without necessarily worrying about verification. 

For the maximum cost, the optimal solution is clearly the midpoint between the leftmost and rightmost reported locations. It is interesting to note that the midpoint may not be one of the agents' reported locations --- so, to compute the optimal solution, our expressive representation of the leaves as convex combinations of points in $\vx$ is required. Procaccia and Tennenholtz~\cite{PT13ec} have shown that any truthful mechanism cannot achieve an approximation ratio smaller than $2$ for the maximum cost. A ratio of $2$ is achieved by any solution that places the facility between the leftmost and rightmost reported locations. It follows that the optimal ratio is trivial to obtain truthfully, e.g., by always selecting the location $x_1$ reported by agent $1$. This mechanism is representable via a tiny decision tree with one leaf. 

We conclude that, in the context of deterministic mechanisms and the maximum cost objective, truthful mechanisms that are efficiently verifiable can do just as well as any truthful mechanism. 

Let us therefore focus on the social cost. For any number of agents $n$, it is easy to see that selecting the median of the reported locations is the optimal solution. Indeed, if the facility moves right or left, the facility would get further away from a majority of locations, and closer to a minority of locations. The median mechanism was observed by Moulin~\cite{Moul80} to be truthful. Intuitively, this is because the only way an agent can manipulate the median's location is by reporting a location that is on the other side of the median --- but that only pushes the median away from the agent's actual location. 
Moreover, the median can be computed by a decision tree in which every internal node contains comparisons between the input locations, and each leaf $\mathcal{L}$ outputs the location of the facility (the median)
%(i.e. the median given the constraints on the path to $\mathcal{L}$) 
when $\mathcal{L}$ is reached.

In contrast to the maximum cost, though, the optimal mechanism for the social cost --- the median --- requires a huge decision tree representation. The number of comparisons required to compute the median has been formally studied (see, e.g., Blum et al.~\cite{BFPR+73}), but, in our case, simple intuition suffices: if there is an odd number of agents with distinct locations, the median cannot be determined when nothing is known about the location of one of the agents, so $(n-1)/2$ comparisons are required \emph{in the best case}, leading to a tall binary tree of exponential size. 

%Recall that the median location achieves the optimal social cost~\cite{Mou80}.  If $n = 2k+1$, for some $k$, then the unique optimal point is the median. If $n = 2k$, then given input $\vx$, any location in the interval $[x_{k}, x_{k+1}]$ is optimal; in particular, the mechanism would have to return the point $x_k$ whenever $x_k = x_{k+1}$.
%However, computing the $k$-th order statistic requires at least $n$ comparisons for every $k$, and so any decision tree mechanism with optimal social cost has size exponential in $n$~\cite{BFPR+73}.
%
%
%{\color{red}:Ariel:can you take a look at the next theorem statements -- I'm phrasing them "for all large enough n"}
%
%Next we study the approximation ratios of polynomial-size trees with respect to the social cost.

Our next result strengthens this insight by giving a lower bound on the approximation ratio achievable by polynomial size decision trees (i.e., trees efficiently verifiable by our algorithm of \S\ref{sec:step2}).
%(i.e., trees that can efficiently verified by our algorithm of \S\ref{sec:step2}).

\begin{theorem}
\label{thm:step3lb}
For every constant $k \in \mathbb{N}$, every truthful deterministic decision tree for $n$ agents of size
at most $n^k$ has an approximation ratio of $\Omega\left(\frac{n}{\log{n}}\right)$ for the social cost.
\end{theorem}

On the positive side, we show that the lower bound of Theorem~\ref{thm:step3lb} is asymptotically tight. 

\begin{theorem} \label{thm:det_apx}
For every $n\in\mathbb{N}$ there is a truthful deterministic decision tree of size  $O(n^6)$ that approximates
the social cost within a factor of $O\left(\frac{n}{\log(n)}\right)$.
\end{theorem}

\begin{comment}
I also think this is in fact the best possible among deterministic decision trees, at least for mechanisms that select a location among the input values x = (x_1, ..., x_n):
Construct a set of input instances S, where for every permutation \pi of the set of agents, the set S contains two possible inputs (each with relative order \pi):
1). x = { k points clustered around 0 (say within epsilon distance) and n - k points around 1 }
2). x = { n - k points clustered around 0 and k around 1 }
\end{comment}

%\begin{todo}
%Add statement about the max cost objective.
%\end{todo}

In summary, polynomial-size decision trees can achieve the best possible approximation ratio
(among all truthful deterministic mechanisms) with respect to the maximum cost objective and an approximation ratio of $\Theta(n/\log n)$ with respect to the social cost.

\subsection{Randomized Mechanisms}

We next turn to randomized mechanisms. In this context, we are interested in the expected social cost, or the expected maximum cost. The latter measure is somewhat subtle, so let us state specifically that, like Procaccia and Tennenholtz~\cite{PT13ec}, we are interested in 
$$
\mathbb{E}\Big[\mac(\vx,\mathcal{M}(\vx))\Big]=\mathbb{E}\left[\max_{i\in N}\cost(x_i,\mathcal{M}(\vx))\right].
$$
A less stringent alternative would be to take the maximum over agents of the agent's expected cost.  

It is immediately apparent that universally truthful, randomized, small decision trees can easily beat the lower bound of Theorem~\ref{thm:step3lb} for social cost. To see this, consider the random dictator mechanism, that selects an agent $i\in N$ uniformly at random, and returns the location $x_i$. This mechanism is clearly universally truthful (it is a uniform distribution over dictatorships), and it is easy to verify that its approximation ratio is $2-2/n$. 

Our next theorem, which we view as the main result of this section, shows that randomization allows us to get arbitrarily close to $1$ using universally truthful, efficiently-verifiable mechanisms.

\begin{theorem}
\label{thm:step3rand}
For every $0 < \epsilon < \frac{1}{10}$ and $n\in \mathbb{N}$, 
there exists a universally truthful randomized decision tree of polynomial size in $n$ that
%size $O(poly(n))$,
%where $t = \lceil \frac{100\log{2n}}{\epsilon^2}\rceil$, 
approximates
the social cost to a factor of $1+ \epsilon$.
%$- 1 + \delta + (1 - \delta)\frac{2}{1 - 2 \epsilon} + \delta(n-1)$.
%\gamma \cdot \frac{7}{\epsilon^{2}} \log\left(\frac{2}{\delta}\right)})$ and approximates the optimal social cost within a factor of $\frac{2(1- \delta)}{ 1 - 2 \epsilon} + 3 \delta -1 $, where $\gamma > 0$ is a constant.
\end{theorem}

In stark contrast, universal truthfulness does not help obtain a better bound than the trivial approximation ratio of $2$ for the maximum cost --- even in the case of general mechanisms.
%%%%SIMINA: COMMENTED OUT::: decision tree is not restricted in size. We do not know whether the same is true for universally truthful randomized mechanisms that are not necessarily represented as (randomized) decision trees.

\begin{theorem}
\label{thm:jeefa}
For each $\epsilon > 0$, there exists no universally truthful mechanism given by a distribution over countably many deterministic mechanisms that can approximate the maximum cost within a factor of $2 - \epsilon$.
\end{theorem}

We have the following corollary for universally truthful decision trees.

\begin{corollary}
For each $\epsilon > 0$, there exists no universally truthful decision tree mechanism given by a distribution over countably many deterministic decision trees that can approximate the maximum cost within a factor of $2 - \epsilon$.
\end{corollary}

\begin{comment}
\begin{theorem}
\label{thm:jeefa}
For every $\epsilon > 0$, there exists no universally truthful decision tree that can approximate the maximum cost within a factor smaller than $2 - \epsilon$.
\end{theorem}
\end{comment}

%\vspace{-2mm}
\section{Discussion}
\label{sec:disc}

Theorem~\ref{thm:jeefa} shows that universally truthful decision trees cannot achieve a nontrivial (better than $2$) approximation for the maximum cost. In contrast, Procaccia and Tennenholtz~\cite{PT13ec} designed a \emph{truthful-in-expectation} mechanism that approximates the maximum cost to a factor of $3/2$. This motivates the study of truthful-in-expectation randomized decision trees, as an alternative to universal truthfulness. However, we do not know whether truthfulness in expectation can be efficiently verified (and we believe that it cannot). Intuitively, the main difficulty is that, for every selection of one leaf from each tree in the support of the randomized mechanism, a na\"ive verification algorithm would need to reason about whether a certain location profile $\vx$ can reach this collection of leaves under the constraints imposed by the different trees. 

Our work focuses on the case of locating one facility on the line, which is quite simple from the approximate-mechanism-design-without-money viewpoint. Researchers have investigated approximate mechanism design in generalized facility location settings, involving multiple facilities~\cite{PT13ec,LWZ09,LSWZ10,NST12,FT10,FT13a,FT13b}, different cost functions~\cite{WF13,FT13b}, metric spaces and graphs~\cite{AFPT10,LSWZ10}, and so on. Of these generalizations and extensions, all but one only require a rethinking of our results of \S\ref{sec:step3} --- that is, mechanisms can still be represented as polynomial-size decision trees. But moving from the real line to a more general metric space requires a revision of the way mechanisms are represented in our framework. 

We conclude by re-emphasizing the main message of our paper. In our view, our main contribution is the three-step approach to the design of verifiably truthful mechanisms. Our technical results provide a proof of concept by instantiating this approach in the context of a well-studied facility location setting, and constructing verifiably truthful mechanisms that achieve good quality guarantees. We firmly believe, though, that the same approach is widely applicable. For example, is there a class of mechanisms for combinatorial auctions that gives rise to verifiably truthful mechanisms providing a good approximation to social welfare? One can ask similar questions in the context of every problem studied in algorithmic mechanism design (with or without money).

%\vspace{-2mm}
%\newpage

\bibliographystyle{plain}
%%%\bibliography{abb,ultimate}

\appendix

\section{Step II: Constructing a Verification Algorithm}
\label{sec:missing_pseudocode}

Below we give the pseudocode for \textsc{Build-Leaf-Constraints} (Algorithm~\ref{alg:leaf}), which is used in the main verification algorithm.

\begin{algorithm}[h]
\label{alg:leaf}
\caption{\textsc{Build-Leaf-Constraints($\mathcal{T}$)}}
%\setcounter{AlgoLine}{0}
%\def\NoNumber#1{{\def\alglinenumber##1{}\State #1}\addtocounter{ALG@line}{-1}}%
%\textbf{Procedure} \textit{Build-Leaf-Constraints($T$)}:\\
\KwData{mechanism $\mathcal{T}$}
\KwResult{set of symbolic constraints $\mathcal{C}$;
the location at leaf $\mathcal{L}$ is selected 
on input $\vx$ $\iff$ constraints $\mathcal{C}_{\mathcal{L}}(\vx)$ hold}
%\NoNumber{This line will not have a number!}
%TODO: Figure out what this was meant for.... \textit{// Lazy build, do not repeat if the values have been computed}\\
$\mathcal{C} \leftarrow \emptyset$ //\textit{ Initialize the set of constraints}\\
\ForEach{leaf $\mathcal{L} \in \mathcal{T}$} {
  $Q \leftarrow \mathcal{L}$\\
  \While{$Q \neq Null$}{
      //\textit{ Add the constraint that must hold for $Q$ to be reached from $parent(Q)$}\\
      $c \leftarrow \mbox{\emph{constraint}}(\mbox{\emph{parent}}(Q).\mbox{\emph{Next}}() = Q)$\\
      \Switch{$c$} {
          \Case{$x_{i_c} \geq x_{j_c}$} {
              $\mathcal{C}_{\mathcal{L}}(\vx) \leftarrow \mathcal{C}_{\mathcal{L}}(\vx) \cup \{x_{i_c} - x_{j_c} \geq 0\}$\\
          }
          \Case{$x_{i_c} > x_{j_c}$} {
              $\mathcal{C}_{\mathcal{L}}(\vx) \leftarrow \mathcal{C}_{\mathcal{L}}(\vx) \cup \{x_{i_c} - x_{j_c} \geq 1 \}$ \\
          }
          \Case{$x_{i_c} \leq x_{j_c}$} {
              $\mathcal{C}_{\mathcal{L}}(\vx) \leftarrow \mathcal{C}_{\mathcal{L}}(\vx) \cup \{x_{j_c} - x_{i_c} \geq 0\}$ \\ 
          }
          \Case{$x_{i_c} < x_{j_c}$} {
              $\mathcal{C}_{\mathcal{L}}(\vx) \leftarrow \mathcal{C}_{\mathcal{L}}(\vx) \cup \{x_{j_c} - x_{i_c} \geq 1\}$ \\
          }
      }
      $Q \leftarrow \mbox{\emph{parent}}(Q)$\\
  }
%  \textit{// Build symbolic expression for the facility at $\mathcal{L}: y_{\mathcal{L}}(x) \leftarrow c_{1,\mathcal{L}} \cdot x_1 + \ldots + c_{n,\mathcal{L}} \cdot x_n$}\\
%  $y_{\mathcal{L}}(x) \leftarrow 0$ \\
%  \ForEach{$i \in N$} {
%    $y_{\mathcal{L}}(x) \leftarrow y_{\mathcal{L}}(x) + c_{i,\mathcal{L}} \cdot x_i$
%  }
}
\textbf{return} $\mathcal{C}$
\end{algorithm}

\noindent\textbf{Theorem~\ref{thm:step2lb}} (restated):
\emph{Let $N=\{1,\ldots,n\}$ with $n\geq 2$, and $\ell\leq n!$.
% \ in \left\{1, \ldots, 2^{\binom{n}{2}}\right\}$
Then any algorithm that verifies truthfulness for every deterministic decision tree with $\ell$ leaves for $n$ agents must inspect all the leaves in the worst case.}
\begin{proof}
%We show that given any number of players $n$ and $t \in \{1,\ldots,n!\}$, 
%a verifier that establishes truthfulness for every decision tree of size $t$ must inspect all the leaves, i.e. at least $\lceil \frac{t}{2} \rceil$ nodes.
Assume by contradiction there exists a verification algorithm that can check truthfulness for every tree with $\ell$ leaves without inspecting all the leaves.
Let $\mathcal{T}$ be a decision tree in which every internal node has the form $x_i < x_j$, for $ i,j \in N$ such that $i<j$, and the location is set to $x_1$ in every leaf. Since there are $n!$ possible orders of the agent locations, we can generate such a tree with $\ell$ leaves. 
Clearly, $\mathcal{T}$ is truthful since it coincides
with the mechanism in which agent $1$ is a dictator.

% Moreover, the size of $\mathcal{T}$ can be exponential in the number of agents, $n$, since a decision tree in which every internal node compares two input variables using a strict inequality in lexicographic order of the variable names (i. e. $x_i < x_j$ if $i < j$ and $x_j < x_i$ if $j < i$) can be used to sort the input vector $x$.

Consider the execution of the verification algorithm on input $\mathcal{T}$ and let $\mathcal{L}$ be a leaf that is not inspected by the algorithm. 
Construct a tree $\mathcal{T}'$ that is identical to $\mathcal{T}$, with the exception of leaf
$\mathcal{L}$, where the selected location is $y_{\mathcal{L}}(\vx) = \frac{x_1 + \ldots x_n}{n}$. First note that mechanism $\mathcal{T}'$ is not truthful. 
For every leaf of $\mathcal{T}'$, the mechanism cannot enforce that two variables are equal, since that would require comparing $x_i < x_j$ and $x_j < x_i$ (similarly if weak inequalities are used). However, if $i < j$ then $\mathcal{T}'$ can only check if 
$x_i < x_j$; similarly, if $j < i$, then $\mathcal{T}'$ can only check if $x_j < x_i$.
Thus the leaf $\mathcal{L}$ can be reached when the input $\vx$ is consistent with some strict ordering $\pi$ on $n$ elements.

Define $\vx \in \mathbb{R}^{n}$
such that $x_{\pi_1} < x_{\pi_2} < \ldots < x_{\pi_n}$. Then $y_{\mathcal{L}}(\vx) = \frac{x_1 + \ldots x_n}{n}$ and the cost of agent $\pi_n$ is $\cost(x_{\pi_n}, y_{\mathcal{L}}(\vx)) = x_{\pi_n} - y_{\mathcal{L}}(\vx)$.
There exists $\delta > 0$ such that by reporting $x_{\pi_n}' = x_{\pi_n} + \delta$, agent $\pi_n$ ensures that leaf $\mathcal{L}$ is still reached
%that the facility is still chosen according to $\mathcal{L}$ 
and the new cost is lower:
\begin{eqnarray*}
\cost(x_{\pi_n}, y_{\mathcal{L}}(x_{\pi_n}', \vx_{-\pi_n})) & = & x_{\pi_n} - \frac{\left(\sum_{i \neq \pi_n} x_i\right) + \left(x_{\pi_n} + \delta\right)}{n} \\
& < &  x_{\pi_n} - \frac{x_1 + \ldots + x_n}{n} \\
& = & \cost(x_{\pi_n}, y_{\mathcal{L}}(\vx)).
\end{eqnarray*}
However, since the verification algorithm does not inspect leaf $\mathcal{L}$, it cannot distinguish between $\mathcal{T}$ and $\mathcal{T}'$, and so it decides that $\mathcal{T}'$ is also truthful. This contradicts the correctness of the verification algorithm. 
%Denote by $r$ the index of the 
%
%Let $\pi$ be the ordering of the $x_i$'s required to reach $\mathcal{L}$.
%Given $n$ the number of agents, do I want to show that there exists a truthful mechanism for every tree size $|T|$? Clearly yes if I allow redundancy -- i.e. I make all these checks and then just output the dictator
%$x_1$ everywhere. 
\end{proof}

\section{Step III: Measuring the Quality of Verifiably Truthful Mechanisms}

\subsection{Mising Proofs: Deterministic Mechanisms}

\noindent\textbf{Theorem~\ref{thm:step3lb}} (restated):
\emph{Let $N=\{1,\ldots,n\}$ with $n\geq 2$, and $\ell\leq n!$.
For every constant $k \in \mathbb{N}$, every truthful deterministic decision tree for $n$ agents of size
at most $n^k$ has an approximation ratio of $\Omega\left(\frac{n}{\log{n}}\right)$ for the social cost.
}
\begin{proof}
%We show there exists a data set $D$ such that every deterministic mechanism makes a mistake with a penalty of $\Omega\left(\frac{n}{\log(n)}\right)$ on at least one input from $D$. Define 
%\[
%D = \left\{ x = (x_1, \ldots, x_n) \in \{0,1\}^{n} \; | \; \exists S \subseteq N \; \mbox{such that} \; x_i = 0 \; \mbox{if} \; i \in S \; \mbox{and} \; x_i = 1 \; \mbox{otherwise}\right\}
%\]
%
%
Let $\mathcal{M}$ be a deterministic mechanism represented by some decision tree $\mathcal{T}$ of size at most $n^k$.
Recall that every internal node in $\mathcal{T}$ checks the order of two input variables with one of the following inequalities: $\{x_i \geq x_j$, $x_i \leq x_j$, $x_i < x_j$, $x_i > x_j\}$. 

Since $\mathcal{T}$ is binary and $|\mathcal{T}| \leq n^k$, there exists at least one leaf $\mathcal{L} \in \mathcal{T}$ of depth 
$$d < 2 \cdot \log(|\mathcal{T}|) \leq 2 \log(n^k) = 2k\log(n).$$
Let $S_{\mathcal{L}} = \{i_1, \ldots, i_m\}$ be the set of agents whose locations are inspected on the path to $\mathcal{L}$. It holds that $|S_{\mathcal{L}}| = m \leq 2 \cdot d \leq 4 k \cdot \log(n)$, since $\mathcal{L}$ has depth $d$ and each node on the path to $\mathcal{L}$ inspects two locations. Note that if $S_{\mathcal{L}} = \emptyset$, then $\mathcal{M}$ is a dictatorship, and so its approximation ratio is no better than $n-1$. Thus we can assume that $S_{\mathcal{L}} \neq \emptyset$.

Recall that the facility at $\mathcal{L}$ is a convex combination of the input locations; that is, $y_{\mathcal{L}}(\vx) = \sum_{i=1}^{n} \lambda_{\mathcal{L},i} \cdot x_i$, where $\lambda_{\mathcal{L},i} \in [0,1], \forall i \in N$ 
and $\sum_{i=1}^{n} \lambda_{\mathcal{L},i} = 1$. Let $\pi$ be a weak ordering consistent with the leaf $\mathcal{L}$ and
$D_{\mathcal{L}} = \{i_1, \ldots, i_l\}$ a ``deduplicated'' version of $S_{\mathcal{L}}$, such that $D_{\mathcal{L}}$ contains one representative agent $i$
for each maximal subset $W \subseteq S_{\mathcal{L}}$ with the property that under $\pi$, $x_j = x_i$, $\forall j \in W$.
Note that $D_{\mathcal{L}}$ is consistent with some strict ordering $\sigma$ on $l$ elements. 

We distinguish between three cases:

\begin{enumerate}
%\item The facility is a convex combination of locations in $S_{\mathcal{L}}$ only. That is, $\lambda_{\mathcal{L},i} = 0, \forall i \not \in S_{\mathcal{L}}$. 
%Let $D_{\mathcal{L}} = \{i_1, \ldots, i_l\}$ be a deduplicated version of $S_{\mathcal{L}}$, such that $D_{\mathcal{L}}$ contains one representative agent
%for each maximal subset $W \subset S_{\mathcal{L}}$ with the property that $x_i = x_j$, $\forall i,j \in W$.
\item The facility at $\mathcal{L}$ is a convex combination of agents in $S_{\mathcal{L}}$ only (i.e., $\lambda_{\mathcal{L},i} = 0$, $\forall i \not \in S_{\mathcal{L}}$).

Let $\epsilon$ be fixed such that $0 < \epsilon < \frac{|S_{\mathcal{L}}|}{n}$ and define the following input $\vx = \langle x_1, \ldots, x_n \rangle$:
\begin{itemize}
\item For each $i \in D_{\mathcal{L}}$, let $r_i$ be the number of agents in $D_{\mathcal{L}}$ strictly to the left of $i$ according to $\sigma$; set $x_i \leftarrow \epsilon \cdot \left(\frac{r_i}{n}\right)$.
\item For each $j \in S_{\mathcal{L}} \setminus D_{\mathcal{L}}$, set $x_j \leftarrow x_i$, where $i \in D_{\mathcal{L}}$ and $x_i = x_j$ according to $\pi$.
\item For each $j \not \in S_{\mathcal{L}}$, set $x_j \leftarrow 1$.
\end{itemize}
The optimal location of the facility given $\vx$ is $y^{*} = 1$, since most agents are situated at $1$ (except the agents in $S_{\mathcal{L}}$, 
of which there are at
most: $4k\log(n) \ll n/2$). The optimal social cost is:
\[
\soc^{*}(\vx) = \sum_{i = 1}^{n} \cost(x_i, y^*) = \sum_{i \in S_{\mathcal{L}}} (1 - x_i) \leq 1 \cdot |S_{\mathcal{L}}|.
\]

On the other hand, the output of the mechanism is $y_{\mathcal{L}}(\vx) = \sum_{i \in S_{\mathcal{L}}} \lambda_{\mathcal{L},i} \cdot x_i \leq \epsilon$; the 
social cost incurred by $\mathcal{M}$ on $\vx$ is:
\[
\soc(\vx,\mathcal{M}(\vx)) = \sum_{i=1}^{n} \cost(x_i, y_{\mathcal{L}}(\vx)) \geq (n - |S_{\mathcal{L}}|) \cdot (1 - \epsilon).
\]

Choosing $\epsilon\leq 1/n$, the approximation ratio of $\mathcal{M}$ is no better than:
\begin{eqnarray*}
\frac{\soc(\vx,\mathcal{M}(\vx))}{\soc^{*}(\vx)} & \geq & \frac{(n - |S_{\mathcal{L}}|) \cdot (1 - \epsilon)}{|S_{\mathcal{L}}|} \\
& = & \frac{n}{|S_{\mathcal{L}}|} - \frac{n\epsilon}{|S_{\mathcal{L}}|} - 1 + \epsilon \\
& > & \frac{n}{4k \log(n)} - 2 \in \Omega\left(\frac{n}{\log(n)}\right).
\end{eqnarray*}

\item The facility coincides with the location of some agent $t \not \in S_{\mathcal{L}}$ (i.e. $y_{\mathcal{L}}(\vx) = x_t$).

Similarly to Case 1, let $\epsilon$ be fixed such that
$0 < \epsilon < \frac{|S_{\mathcal{L}}| + 1}{n}$
and define 
$\vx = \langle x_1, \ldots, x_n\rangle$ as follows:
\begin{itemize}
\item For each $i \in D_{\mathcal{L}}$, let $r_i$ be the number of agents in $D_{\mathcal{L}}$ strictly to the left of $i$ according to $\sigma$;
%let $r_i$ be the rank of $x_i$ according to $\sigma$; 
set $x_i \leftarrow \epsilon \cdot \left(\frac{r_i}{n}\right)$.
\item For each $j \in S_{\mathcal{L}} \setminus D_{\mathcal{L}}$, set $x_j \leftarrow x_i$, where $i \in D_{\mathcal{L}}$ and $x_i = x_j$ according to $\pi$.
\item Set $x_t = 0$.
\item For each $j \not \in S_{\mathcal{L}}, j \neq t$, set $x_j \leftarrow 1$.
\end{itemize}

The optimal location on $\vx$ is $y^* = 1$, since most agents are located at $1$ (except agent $t$ and the agents in $S_{\mathcal{L}}$).
As in Case 1, by also taking agent $t$ into account, we get:

\begin{eqnarray*}
\frac{\soc(\vx,\mathcal{M}(\vx))}{\soc^{*}(\vx)} & \geq &
\frac{(n - |S_{\mathcal{L}}| - 1) \cdot (1 - \epsilon)}{|S_{\mathcal{L}}| + 1} 
%\\
%& = & \frac{n}{|S_{\mathcal{L}}| + 1} - \frac{n\epsilon}{|S_{\mathcal{L}}| + 1} - 1 + \epsilon \\
%& > & \frac{n}{k(4c+4) \log(n) + 1} - 2 
\in \Omega\left(\frac{n}{\log(n)}\right).
\end{eqnarray*}

\item The facility is a weighted sum with at least two terms, one of which is an agent $t \not \in S_{\mathcal{L}}$. 
We claim that no mechanism that is truthful on the full domain (i.e. the line) 
can have such an output at any leaf. Let $\epsilon, \delta > 0$ be such that
\[
\delta = \frac{1}{2}\left( \frac{1}{\lambda_{\mathcal{L},t}}-1\right) \; \mbox{and} \; \epsilon = \frac{1 - \lambda_{\mathcal{L},t} (1 + \delta)}{n-1}.
\]
Consider an input $\vx$ consistent with the ordering $\pi$ such that $x_t = 1$ and $x_i \in (0, \epsilon), \forall i \neq t$. 
Then:
\[
y_{\mathcal{L}}(\vx) = \sum_{i=1}^{n} \lambda_{\mathcal{L},i} \cdot x_i = 
\left(\sum_{i \neq t} \lambda_{\mathcal{L},i} \cdot x_i \right) + \lambda_{\mathcal{L},t} \cdot 1. %< (n-1) \cdot \epsilon + \lambda_{\mathcal{L},l} \cdot 1
\]
%By deviating to $x_t' = 1 + \delta$, agent $t$ can change the location of the facility to:
If agent $t$ reports instead $x_t' = 1 + \delta$, the output of $\mathcal{M}$ on $\vx' = (x_t', \vx_{-t})$ is:
\[
y_{\mathcal{L}}(\vx') = \left(\sum_{i \neq t} \lambda_{\mathcal{L},i} \cdot x_i \right) + \lambda_{\mathcal{L},t} \cdot (1 + \delta).
\]
It can be verified that $0 < y_{\mathcal{L}}(\vx) < y_{\mathcal{L}}(\vx') < 1$, and so $\cost(x_t, y_{\mathcal{L}}(\vx')) < \cost(x_t, y_{\mathcal{L}}(\vx))$, which contradicts the truthfulness of $\mathcal{M}$.
Thus Case 3 never occurs.
\end{enumerate}

By the cases above, there exists at least one input on which the approximation ratio of $\mathcal{M}$ is $\Omega\left(\frac{n}{\log(n)}\right)$, which completes the proof.
\end{proof}

\noindent\textbf{Theorem~\ref{thm:det_apx}} (restated):
\emph{
For every $n\in\mathbb{N}$ there is a truthful deterministic decision tree of size  $O(n^6)$ that approximates
the social cost within a factor of $O\left(\frac{n}{\log(n)}\right)$.
}
\begin{proof}
First, we claim that for every $k \in \{1, \ldots, n/2\}$,
there exists a truthful, deterministic decision tree of size $O(2^{6k})$ that approximates the social cost within a factor of
$O\left(\frac{n-k}{k}\right)$. Given a fixed $k$, let $\mathcal{M}$ be the following mechanism:
\begin{itemize}
\item Given input $\vx = (x_1, \ldots, x_n)$, output $\mbox{median}(\{x_1, \ldots, x_k\})$.
\end{itemize}
That is, $\mathcal{M}$ always outputs the median of the fixed set of agents $\{1, \ldots, k\}$. Computing the median on an input vector of size $k$ requires fewer than $6k$ comparisons~\cite{BFPR+73}, and since the decision tree for $\mathcal{M}$ is binary, its size is $O(2^{6k})$.

We next claim that the approximation ratio of $\mathcal{M}$ is $O\left(\frac{n-k}{k}\right)$. %, by comparing the location chosen by $\mathcal{M}$ with the location chosen by $OPT$ (the median mechanism).
Indeed, given any instance $\vx \in \mathbb{R}^{n}$, denote $\tilde{m} = \mathcal{M}(\vx)$ and $m^{*} = \text{argmin}_{y\in \mathbb{R}} \soc(\vx,y)$. Without loss of generality, assume that $\tilde{m} < m^{*}$
and let $\Delta = |\tilde{m} - m^{*}|$. Let $S_l = \{x_i \; | \; x_i \leq \tilde{m}\}$, $S_r = \{x_i \; | \; x_i \geq m^{*}\}$, and $S_m = \{ x_i \; | \; \tilde{m} < x_i < m^{*}\}$ be the sets of points to the left of $\tilde{m}$, to the right of $m^{*}$, and strictly between
$\tilde{m}$ and $m^{*}$, respectively. Denote the sizes of the sets by $n_l = |S_l|$, $n_r = |S_r|$, and $n_m = |S_m|$, where $n_l + n_m + n_r = n$.

We compute the upper bound by comparing the social cost of $\mathcal{M}$ on $\vx$, $\soc(\vx,\mathcal{M}(\vx)) = \sum_{i=1}^{n} \cost(x_i, \tilde{m})$, with $\soc^{*}(\vx) = \sum_{i=1}^{n} \cost(x_i, m^{*})$.
Observe that for all the points in $S_{r}$, the cost increases by exactly $\Delta$ when moving the location from $m^{*}$ to $\tilde{m}$. On the other hand, the change from $m^{*}$ to $\tilde{m}$ results in a decrease by exactly $\Delta$ for
the points in $S_{l}$.
% the cost decreases by exactly $\Delta$
Thus $\soc(\vx,\mathcal{M}(\vx))$ can be expressed as follows:
\[
\soc(\vx,\mathcal{M}(\vx)) = \soc^{*}(\vx) + n_r \cdot \Delta + \sum_{j \in S_m} [\cost(x_j, \tilde{m}) - \cost(x_j, m^{*})] - n_l \cdot \Delta.
\]
The ratio of the costs is:
\[
\frac{\soc(\vx,\mathcal{M}(\vx))}{\soc^*(\vx)}= \frac{\soc^*(\vx) + n_r \cdot \Delta + \sum_{j \in S_m} [\cost(x_j, \tilde{m}) - \cost(x_j, m^{*})] - n_l \cdot \Delta}{\soc^*(\vx)}.
\]
We claim that 
\begin{equation} \label{eq1}
\frac{\soc(\vx,\mathcal{M}(\vx))}{\soc^*(\vx)} \leq \frac{3(n-k)}{k}.
\end{equation}
Inequality (\ref{eq1}) is equivalent to:
\[
k \cdot n_r \cdot \Delta + k \cdot \sum_{j \in S_m} [\cost(x_j, \tilde{m}) - \cost(x_j, m^{*})] - k \cdot n_l \cdot \Delta \leq (3n - 4k) \soc^*(\vx).
\]
Note that for all $j \in S_m$, $\cost(x_j, \tilde{m}) - \cost(x_j, m^{*}) \leq \Delta$, and so if Inequality (\ref{eq1}) holds when 
$\cost(x_j, \tilde{m}) - \cost(x_j, m^{*}) = \Delta$, then it also
holds for all other instances where the change in cost is smaller for some agents $j \in S_m$. Formally, if:
\begin{equation} \label{eq2}
k \cdot n_r \cdot \Delta + k \cdot n_m \cdot \Delta - k \cdot n_l \cdot \Delta \leq (3n - 4k) \soc^*(\vx),
\end{equation}
then Inequality (\ref{eq1}) also holds. Inequality (\ref{eq2}) is equivalent to:

\begin{equation}
\label{eq3}
  \begin{alignedat}{2}
\soc^*(\vx) & \; \geq \; && \frac{k \cdot n_r \cdot \Delta + k \cdot n_m \cdot \Delta - k \cdot n_l \cdot \Delta}{3n - 4k} \\
&\; = \; && \frac{k \cdot (n_r + (n - n_l - n_r) - n_l) \cdot \Delta}{3n - 4k}\\
& \; = \; && \frac{k \cdot (n - 2n_l) \cdot \Delta}{3n - 4k}.
  \end{alignedat}
\end{equation}

Each of the agents in $S_l$ pays a cost of at least $\Delta$ under $\soc^*(\vx)$, and so $\soc^*(\vx) \geq n_l \cdot \Delta$. Moreover, since $\tilde{m}$ is the median of $\{x_1, \ldots, x_k\}$, it follows that $n_l \geq \frac{k}{2}$.
We first show that $n_l \cdot \Delta \geq \frac{k \cdot (n - 2n_l) \cdot \Delta}{3n - 4k}$:
%which combined with $C_{OPT} \geq n_l \cdot \Delta$, will give the condition required by the previous inequality 
%\begin{small}
\begin{equation} 
\label{eq4}
\begin{split}
& n_l \cdot \Delta \geq \frac{k \cdot (n - 2n_l) \cdot \Delta}{3n - 4k}  \\
 \iff & n_l (3n - 4k) \geq k(n - 2n_l) \\
 \iff  & n_l (3n - 2k) \geq kn  \\
\iff & n_l \geq \frac{kn}{3n - 2k}
\end{split}
\end{equation}
%\end{small}
In addition, we have that
\begin{equation} \label{eq5}
\frac{k}{2} \geq \frac{kn}{3n - 2k} \iff 3kn - 2k^2 \geq 2kn \iff n \geq 2k.
\end{equation}
Inequality (\ref{eq5}) holds by the choice of $k$; combining it with $n_l \geq \frac{k}{2}$, we obtain: 
\begin{equation}\label{eq6}
n_l \geq \frac{k}{2} \geq \frac{kn}{3n - 2k}.
\end{equation}
By Inequality \eqref{eq3}, it follows that $n_l \cdot \Delta \geq \frac{k \cdot (n - 2n_l) \cdot \Delta}{3n - 4k}$. In addition, $\soc^*(\vx) \geq n_l \cdot \Delta$, thus:
\[
\soc^*(\vx) \geq n_l \cdot \Delta \geq \frac{k \cdot (n - 2n_l) \cdot \Delta}{3n - 4k}.
\]
Equivalently, Inequality (\ref{eq2}) holds, which gives the worst case bound required for Inequality (\ref{eq1})
to always hold. Thus $\frac{\soc(\vx,\mathcal{M}(\vx))}{\soc^*(\vx)} \leq \frac{3(n-k)}{k}$, for every input $\vx$.

Let $k = \log n$. Then $\mathcal{M}$ can be implemented using a decision tree of size 
$O(n^6)$ and has an approximation ratio bounded by 
\[
\frac{\soc(\vx,\mathcal{M}(\vx))}{\soc^*(\vx)} \leq \frac{3(n-k)}{k} 
%= \frac{3n - \frac{\log(n)}{2}}{\left(\frac{\log(n)}{6}\right)} 
\in O\left(\frac{n}{\log(n)}\right)
\]
This completes the proof of the theorem.
\end{proof}

\subsection{Missing Proofs: Randomized Mechanisms}

\noindent\textbf{Theorem~\ref{thm:step3rand}} (restated):
\emph{
For every $0 < \epsilon < \frac{1}{10}$ and $n\in \mathbb{N}$, 
there exists a universally truthful randomized decision tree of polynomial size in $n$ that
%size $O(poly(n))$,
%where $t = \lceil \frac{100\log{2n}}{\epsilon^2}\rceil$, 
approximates
the social cost to a factor of $1+ \epsilon$.
%$- 1 + \delta + (1 - \delta)\frac{2}{1 - 2 \epsilon} + \delta(n-1)$.
%\gamma \cdot \frac{7}{\epsilon^{2}} \log\left(\frac{2}{\delta}\right)})$ and approximates the optimal social cost within a factor of $\frac{2(1- \delta)}{ 1 - 2 \epsilon} + 3 \delta -1 $, where $\gamma > 0$ is a constant.
}
The idea is the following: we sample a subset of agents of logarithmic size -- more exactly $O\left(\frac{\ln(n/\epsilon)}{\epsilon^2}\right)$ -- and select the median among their reported locations. To reason about this mechanism, we define the rank of an element $x$ in a set $S$ ordered by $\succ$ to be $\rank(x)=|\{y\in S\mid y\succ x \vee y=x\}|$, and the \emph{$\epsilon$-median} of $S$ to be $x\in S$ such that $(1/2-\epsilon)|S|<\rank(x)<(1/2+\epsilon)|S|$. The following lemma is a folklore result when sampling is done with replacement; we include its proof because we must sample without replacement. 

\begin{lemma} \label{lem:sampling}
Consider the algorithm that samples $t$ elements without replacement from a set $S$ of cardinality $n$, and returns the median of the sampled points. For any $\epsilon,\delta<1/10$, if 
$$
\frac{100\ln\frac{1}{\delta}}{\epsilon^2}\leq t\leq \epsilon n,
$$
then the algorithm returns an $\epsilon$-median with probability $1-\delta$.  
\end{lemma}

\begin{proof}
We partition $S$ into three subsets:
$$
S_1=\{x\in S\mid \rank(x)\leq n/2-\epsilon n\},
$$ 
$$S_2=\{x\in S \mid n/2-\epsilon n<\rank(x)< n/2+\epsilon n\},
$$
and 
$$
S_3=\{x\in S\mid \rank(x)\geq n/2-\epsilon n\}.
$$ 
Suppose that $t$ elements are sampled without replacement from $S$. If less than $t/2$ are sampled from $S_1$, and less than $t/2$ are sampled from $S_3$, then the median of the sampled elements will belong to $S_2$ --- implying that it is an $\epsilon$-approximate median. 

Let us, therefore, focus on the probability of sampling \emph{at least} $t/2$ samples from $S_1$. Define a Bernoulli random variable $X_i$ for all $i=1,\ldots,t$, which takes the value 1 if and only if the $i$'th sample is in $S_1$. 

Note that $X_1,\ldots,X_t$ are not independent (because we are sampling with replacement), but for all $i$ it holds that 
$$
\Pr[X_i=1\mid X_1=x_1,\ldots,X_{i-1}=x_{i-1}]\leq \frac{\frac{n}{2}-\epsilon n}{n-(i-1)}\leq \frac{\frac{n}{2}-\epsilon n}{n-\epsilon n}\leq \frac{1}{2}-\frac{\epsilon}{3},
$$
for any $(x_1,\ldots,x_{i-1})\in\{0,1\}^{i-1}$, where the second inequality follows from $i\leq t\leq \epsilon n$. 

Let $Y_1,\ldots,Y_t$ be i.i.d. Bernoulli random variables such that $Y_i=1$ with probability $1/2-\epsilon/3$. Then for all $x$, 
$$
\Pr\left[\sum_{i=1}^t X_i\geq x\right]\leq \Pr\left[\sum_{i=1}^t Y_i\geq x\right].
$$
Using Chernoff's inequality, we conclude that
\begin{align*}
\Pr\left[\sum_{i=1}^t X_i\geq \frac{t}{2}\right]&\leq \Pr\left[\sum_{i=1}^t Y_i\geq \frac{t}{2}\right]
=\Pr\left[\sum_{i=1}^t Y_i\geq \left(1+\frac{\epsilon}{\frac{3}{2}-\epsilon}\right)\mathbb{E}\left[\sum_{i=1}^t Y_i  \right]\right]\\
&\leq\Pr\left[\sum_{i=1}^t Y_i\geq \left(1+\frac{\epsilon}{2}\right)\mathbb{E}\left[\sum_{i=1}^t Y_i  \right]\right]\leq \text{exp}\left(-\frac{\left(\frac{\epsilon}{2}\right)^2\left(\frac{1}{2}-\frac{\epsilon}{3}\right)t}{3}\right)\leq \frac{\delta}{2},
\end{align*}
where the last inequality follows from the assumption that $t\geq \frac{100\ln(1/\delta)}{\epsilon^2}$. The lemma's proof is completed by applying symmetric arguments to $S_3$, and using the union bound.
\end{proof}

\begin{proof}[Proof of Theorem~\ref{thm:step3rand}]
Let $\vx = \langle x_1, \ldots, x_n\rangle$ be the set of inputs.
For every $k \in N$, define the mechanism $\mathcal{M}_{n,k}$ as follows:
\begin{itemize}
\item Select uniformly at random a subset $S_k \subseteq N$, where $|S_k| = k$.
\item Output the median of $S_k$.
\end{itemize}
Note that $\mathcal{M}_{n,1}$ coincides with random dictator, while $\mathcal{M}_{n,n}$ is the median mechanism.
Recall that random dictator, $\mathcal{M}_{n,1}$, has an approximation ratio of $2 - 2/n$ for the social cost, while the median, $\mathcal{M}_{n,n}$, is optimal.
The approximation ratio of $\mathcal{M}_{n,k}$ improves as $k$ grows from $1$ to $n$ and the mechanism is universally truthful for every  $k$; in particular, we show there exists a choice of $k$ that achieves a good tradeoff between the size of the mechanism and its approximation ratio.

First, we describe the implementation of $\mathcal{M}_{n,k}$ as a randomized decision tree.
The root has outgoing degree one and selects a function $\mathcal{F}$ that takes
$k$ arguments $Z = \{z_1, \ldots, z_k\}$ and computes the median of $z_1,\ldots,z_k$.
At execution time, $z_1,\ldots,z_k$ are instantiated using the locations $x_{i_1},\ldots,x_{i_k}$ of $k$ distinct agents, chosen uniformly at random from $k$-subsets of $N$. Note that $\mathcal{F}$ can be 
implemented with a decision tree of size $O(2^{6k})$.
%
%(i.e. each $z_{j}$ is mapped to a distinct variable $x_{i_j} \in X$).
%Note that the median on $k$ elements requires fewer than $6n$ comparisons~\cite{BFPR+73}, and so $\mathcal{F}$ can be implemented with a decision tree of size $O(2^{6k})$.

Let $\epsilon', \delta > 0$ be fixed such that $\epsilon', \delta < \frac{1}{10}$.
By Lemma~\ref{lem:sampling}, the algorithm that samples without replacement $t = \lceil \frac{100\ln\frac{1}{\delta}}{(\epsilon')^2}\rceil$ elements 
from a set of $n$ elements returns an $\epsilon'$-median with probability $1- \delta$, as long as $t\leq \epsilon' n$.

Let $\vx\in\mathbb{R}^n$; without loss of generality $x_1\leq \cdots \leq x_n$. We wish to compare $\mathbb{E}[\soc(\vx,\mathcal{M}_{n,t}(\vx))]$ and $\soc^*(\vx)$. Let us suppose that $\mathcal{M}_{n,t}$ returns an $\epsilon'$-median, call it $x_l$. Since $x_l$ is an $\epsilon'$-median, we have that $\frac{n}{2} - \epsilon' n < l < \frac{n}{2} + \epsilon' n$. 
%\ariel{You seem to be assuming wlog that $x_1\leq\cdots\leq x_n$} 
Take the case where $l < \frac{n}{2}$ (the other case, where $l > \frac{n}{2}$, is similar) and let $\Delta = |x_l - x_m|$, where $x_m = \mbox{median}(\vx)$.
Then by moving the facility from $x_m$ to $x_l$, the costs of the agents change as follows:
\begin{enumerate}[(i)]
\item Each agent to the left of $x_l$ (including agent $l$) has the cost decreased by exactly $\Delta$.
\item Each agent strictly between $x_l$ and $x_m$ incurs an increase in cost of at most $\Delta$.
\item Each  agent to the right of $x_m$ (including agent $m$) has the cost increased by $\Delta$.
\end{enumerate}

It follows that 
$$
\soc(\vx,x_l) \leq \soc^*(\vx) - l \cdot \Delta + (n -l) \cdot \Delta = \soc^*(\vx) + (n-2l) \cdot \Delta.
$$ 
On those instances where $\mathcal{M}_{n,t}$ does not return the median, 
the social cost is at most $(n-1) \cdot \diam(\vx)$, where $\diam(\vx) = \max_{i,j\in N} |x_i - x_j|$. 
On the other hand, the optimal cost satisfies the inequalities: $\soc^*(\vx) \geq \diam(\vx)$ and
$\soc^*(\vx) \geq l \cdot \Delta$.

Since $\mathcal{M}_{n,t}$ returns an $\epsilon'$-median with probability $1- \delta$, the ratio of the costs can be bounded by:

\begin{eqnarray*}
\frac{sc_{\mathcal{M}_{n,t}}(\vx)}{\soc^*(\vx)} & \leq & \frac{(1 - \delta) \soc^*(\vx) + \Delta (1 - \delta) (n-2l) + \delta (n-1) \cdot \diam(\vx)}{\soc^*(\vx)} \\
& \leq & (1 - \delta) + \frac{\Delta(1 - \delta)(n - 2l)}{\Delta \cdot l} + \frac{\delta (n-1) \cdot \diam(\vx)}{\diam(\vx)}\\
%& \leq & 1 - \delta + \frac{2 \epsilon' n(1 - \delta)}{(\frac{n}{2}-\epsilon' n)} + 2 \delta \\
& = & 1 - \delta + (1 - \delta)\frac{n}{l} - 2 (1-\delta) + \delta(n-1)\\
& \leq & \delta \cdot n - 1 + (1 - \delta)\frac{2}{1 - 2 \epsilon'}\leq 1+\delta\cdot n+5\epsilon'.
%& = & -1 + \frac{1}{2n} + \left(1 - \frac{1}{2n}\right)\frac{1}{1 - 2 \epsilon'} + \frac{1}{2n}(n-1)\\
%& = & \left(1 - \frac{1}{2n}\right)\frac{1}{1 - 2 \epsilon'} - \frac{1}{2}
%%%%%%%%& = & \frac{(1 - \delta) \cdot n}{l} + 3 \delta - 1 \\
%& \leq & \frac{(1 - \delta) \cdot n}{\frac{n}{2} - \epsilon'  n} + 3 \delta - 1\\
%& \leq & \frac{2(1- \delta)}{ 1 - 2 \epsilon'} + 3 \delta -1 
\end{eqnarray*}
%Note that the decision tree of $\mathcal{M}(n,t)$ has size polynomial in $n$, since $t \in O(\log{n})$ and the median decision tree 
%for a set of size $t$
%is $O(2^{6c \log{n}})$ for some constant $c > 0$, which is $O(poly(n))$.

Given $\epsilon<1/10$, let $\epsilon'=\epsilon/10$ and $\delta = \epsilon/(2n)$, and set $t = \lceil \frac{100\ln{\frac{1}{\delta}}}{(\epsilon')^2}\rceil$. Then $\mathcal{M}_{n,t}$ can
be represented as a randomized decision tree of
size $O(2^{6t})$, which is polynomial in $n$. Moreover, for this choice of $\epsilon', \delta$, the approximation ratio of $\mathcal{M}_{n,t}$ is bounded by 
$$
1+\delta\cdot n+5\epsilon' = 1=\frac{\epsilon}{2}+\frac{\epsilon}{2}=1+\epsilon.
$$
\end{proof}

\noindent\textbf{Theorem~\ref{thm:jeefa}} (restated):
\emph{
For each $\epsilon > 0$, there exists no universally truthful mechanism given as a distribution over countably many deterministic mechanisms
that can approximate the maximum cost on the line within a factor of $2 - \epsilon$.
}
\begin{proof}
We use the following characterization due to Moulin~\cite{Moul80}. Let a voting scheme be defined as a mapping $\pi : \mathbb{R}^n \rightarrow \mathbb{R}$, such that for every tuple of inputs $\vec{x} = \langle x_1, \ldots, x_n \rangle \in \mathbb{R}^n$, the selected alternative is $\pi(\vec{x}) \in \mathbb{R}$.

\begin{lemma}[Moulin 1980] \label{lem:char}
The voting scheme $\pi$ among $n$ agents is strategy-proof if and only if there
exists for every subset $S \subseteq \{1, \ldots, n\}$ (including the empty set) a real
number $a_S \in \mathbb{R} \cup \{\pm \infty\}$ such that:
\begin{itemize}
\item For each $\vec{x} \in \mathbb{R}^n$, $\pi(\vec{x}) = \inf_{S \subseteq \{1, \ldots, n\}} \Big[ \sup_{i \in S} \{ x_i, a_S\} \Big]$.
\end{itemize}
\end{lemma}

Note that by definition, the output of the mechanism is always finite. This simply restricts the values of $a_S$ such that it cannot be the case that either $(i)$ $a_{\emptyset} = -\infty$, or $(ii)$ $a_S = +\infty$ for all $S \subseteq N$. 
%%%\ariel{Shouldn't (i) be just for $S=\emptyset$?}
%%%SIMINA: YES, fixed.

To get some intuition first, we show an example of the median mechanism in the format required by Lemma~\ref{lem:char}.
Let $N = \{1, 2, 3\}$ and define $a_{\emptyset} = a_{1} = a_{2} = a_{3} = +\infty$ and $a_{12} = a_{23} = a_{31} = a_{123} = - \infty$, where $a_{ij}$ is the constant corresponding to subset $S = \{i,j\}$.
Then for any $\vec{x} \in \mathbb{R}^3$, we have:
\begin{eqnarray*}
\pi(\vec{x}) & = & \inf\Big\{ \sup\{a_{\emptyset}\},\\
& &  \; \; \; \; \; \; \; \; \sup\{x_1, a_1\}, \sup\{x_2, a_2\}, \sup\{x_3, a_3\},\\
& &  \; \; \; \; \; \; \; \; \sup\{x_1, x_2, a_{12}\}, \sup\{x_2, x_3, a_{23}\}, \sup\{x_3, x_1, a_{13}\},\\
& &  \; \; \; \; \; \; \; \; \sup\{x_1, x_2, x_3, a_{123}\}\Big\}
\end{eqnarray*}

For example, if $x_1 = 1, x_2 = 3$, and $x_3 = 2$, the location of the facility is:
\begin{eqnarray*}
\pi(\langle 1, 3, 2 \rangle) & = & \inf \Big\{ \sup\{+\infty \},\\
& & \; \; \; \; \; \; \; \; \sup\{1, +\infty\}, \sup\{3, +\infty\}, \sup\{2,+\infty\}, \\
& &  \; \; \; \; \; \;  \; \; \sup\{1,3, -\infty\}, \sup\{3,2, -\infty\}, \sup\{2,1, -\infty\},\\
& &  \; \; \; \;  \; \;  \; \; \sup\{1, 3, 2,  -\infty\}\Big\}\\
& = & 2,
\end{eqnarray*}
which represents the median of the input vector.

We can now analyze the approximation ratios of universally truthful mechanisms with respect to the maximum cost objective. Let $\epsilon > 0$. Take any universally truthful mechanism
$\mathcal{M}$, represented as a probability distribution over deterministic truthful mechanisms chosen from a universe
$\mathcal{U} = \{\mathcal{M}_k \; | \; k \in K\}$, where $K \subseteq \mathbb{N}$. Denote by $p_k$ the probability that mechanism $\mathcal{M}_k$ is selected during the execution of $\mathcal{M}$ 
(for any input $\vec{x} \in \mathbb{R}^n$). Note that $p_k > 0$, since otherwise $\mathcal{M}_k$ can be eliminated from the description of $\mathcal{M}$, and $\sum_{k \in K} p_k = 1$.

For each $t \in \mathbb{N}$, define the following:
\begin{itemize}
\item $\mathcal{S}_t = \{ k \; | \; k \in K \; \mbox{and} \; p_k \geq \frac{1}{2^t} \}$ --- the set of indices of mechanisms $\mathcal{M}_k$ taken with probability at least $1/2^t$, and
\item $q_t = \sum_{k \in \mathcal{S}_t} p_k$ --- the probability that some mechanism $\mathcal{M}_k$ with $k \in \mathcal{S}_t$ is selected.
\end{itemize}
Note that $\emptyset \subseteq \mathcal{S}_0 \subseteq \mathcal{S}_1 \subseteq \ldots \subseteq \mathcal{S}_T \subseteq \ldots \subseteq \mathcal{U}$ and $0 \leq q_0 \leq q_1 \leq \ldots \leq q_T \leq \ldots \leq 1$. 
We have that
$\lim_{t \to \infty} q_t = 1$, and so there exists $T \in \mathbb{N}$ such that $q_T > 1 - \epsilon/2$. Clearly $\mathcal{S}_T$ is a finite set and each (deterministic) mechanism $\mathcal{M}_k$ with $k \in \mathcal{S}_T$ has the property that $p_k \geq 1/2^T$.

By Lemma~\ref{lem:char}, for each $k \in \mathcal{S}_T$, there exist constants $a_S^{k} \in \mathbb{R} \cup \{\pm \infty\}$ for each subset $S \subseteq \{1, \ldots, n\}$, such that
$\mathcal{M}_k(\vec{x}) = \inf_{S \subseteq \{1, \ldots, n\}} \Big[ \sup_{i \in S} \{ x_i, a_S^k\} \Big]$, for all $\vec{x} \in \mathbb{R}^n$.
%%%Recall that by definition $\mathcal{M}_k$ always has a finite real output.

%%%Let $p^* = \min_{k \in \{1, \ldots, K\}} p_k$,
%%%where $p_k$ is the probability that $\mathcal{M}$ selects mechanism $\mathcal{M}_k$ on any given input.
%%%Without loss of generality $p^* > 0$, since otherwise there exists a mechanism $\mathcal{M}_k$ that is never executed and so can be eliminated from the description of $\mathcal{M}$.

Define $P = \bigcup_{k \in \mathcal{S}_T} \bigcup_{S \subseteq \{1, \ldots, n\}} \{ a_S^k \; | \; - \infty <  a_S^k < + \infty \}$ as the set of all finite 
constant points hardcoded in the mechanisms indexed by $\mathcal{S}_T$.
Since $P$ is finite, there exists a contiguous interval $[a, a+1]$ on the line such that $P \cap [a,a+1] = \emptyset$ and the points $a, a + 1$ are \emph{far} from the set $P$, i.e. for all $y \in P$, we have that $d(y, a) > 2^T$ and $d(y, a+1) > 2^T$.

%%%\begin{itemize}
%%%\item $d(y_\ell, a) > 2^{T}$, where $y_\ell = \max \{y \; | \; y < a \; \mbox{and} \; y \in P \}$ 
%%%\ariel{what if $\{y \; | \; y < a \; \mbox{and} \; y \in P %%%\}$ is empty? Can't we just define $[a,a+1]$ so that the distance from every point in this %%%interval to every point in $P$ is more than $2^T$?}
%%%SIMINA:: CHANGED.
%%%\item $d(a + 1, y_r) > 2^{T}$, where $y_r = \min \{y \; | \; y > a + 1 \; \mbox{and} \; y \in P \}$
%%%\end{itemize}

Let $\vec{x}$ be defined as follows:
$x_1 = \ldots = x_{n-1} = a$ and $x_n = a + 1$. The optimal maximum cost on input $\vec{x}$ is $1/2$ and can be obtained by placing the facility at $y^* = a + 1/2$.
We analyze the behavior of $\mathcal{M}$ on input $\vec{x}$ and consider two cases:
\begin{enumerate}
\item If there exists $k \in \mathcal{S}_T$ such that $\mathcal{M}_k(\vec{x}) \in P$, then by definition of $P$ and $\vec{x}$, the approximation ratio of $\mathcal{M}$ can be bounded as follows:

\begin{eqnarray*}
\frac{ \mathbb{E} \Big[ \mac (\vec{x}, \mathcal{M}(\vec{x})) \Big] }{\mac^*(\vec{x})} & = & \frac{ \mathbb{E} \Big[ \max_{i \in N} \cost(x_i, \mathcal{M}(\vec{x})) \Big] }{\mac^*(\vec{x})} \\
& = & 2 \cdot \sum_{k \in \mathcal{U}} \left( p_k \cdot  \max_{i \in N} \cost(x_i, \mathcal{M}_k(\vec{x})) \right) \\
& \geq & 2 \left( \frac{1}{2^T} \right) 2^{T} = 2 \\
\end{eqnarray*}

\item Otherwise, $\mathcal{M}_k(\vec{x}) \not \in P$ for all $k \in \mathcal{S}_T$. Then by Lemma~\ref{lem:char}, for each mechanism $M_k$ with $k \in \mathcal{S}_T$, there exists $i_k \in N$ such that
$\mathcal{M}_k(\vec{x}) = x_{i_k}$. Since $x_i \in \{a, a+ 1\}$ for all $i \in N$, the maximum cost incurred when mechanism $\mathcal{M}_k$ gets selected is $d(a, a+ 1) = 1$.
Then by choice of $\mathcal{S}_T$, the
approximation ratio of $\mathcal{M}$ can be bounded by:

\begin{eqnarray*}
\frac{ \mathbb{E} \Big[ \mac (\vec{x}, \mathcal{M}(\vec{x})) \Big] }{\mac^*(\vec{x})} & = & \frac{ \mathbb{E} \Big[ \max_{i \in N} \cost(x_i, \mathcal{M}(\vec{x})) \Big] }{\mac^*(\vec{x})} \\
& = &2 \cdot \sum_{k \in \mathcal{U}} \left( p_k \cdot  \max_{i \in N} \cost(x_i, \mathcal{M}_k(\vec{x})) \right) \\
& \geq & 2 \cdot \sum_{k \in \mathcal{S}_T} \Big( p_k \cdot 1 \Big) = 2 \cdot q_T \\
& > & 2 (1 - \epsilon/2) = 2 - \epsilon
\end{eqnarray*}
\end{enumerate}

From Cases 1 and 2 we obtain that the approximation ratio of $\mathcal{M}$ on input $\vec{x}$ is worse than $2 - \epsilon$, 
which completes the proof of the theorem.
\end{proof}

\end{document}